\documentclass[conference]{IEEEtran}

\usepackage{amsmath}
\usepackage{amssymb}
\usepackage{amsfonts}
\usepackage{amscd}
\usepackage{amsthm}
\usepackage[noadjust]{cite}
\usepackage{threeparttable}
\usepackage{textcomp}
\usepackage{multirow}

\usepackage{mathrsfs}
\usepackage{bbm}
\usepackage{amsbsy}

\usepackage{psfrag}

\usepackage{paralist}

\usepackage{xfrac}

\newtheorem{theorem}{Theorem}

\theoremstyle{definition}
\newtheorem{example}{Example}
\newtheorem{definition}{Definition}

\makeatletter
\renewcommand*\env@matrix[1][c]{\hskip -\arraycolsep
  \let\@ifnextchar\new@ifnextchar
  \array{*\c@MaxMatrixCols #1}}
\makeatother
\usepackage{graphicx}

\newcommand{\Rb}{\mathbbmss{R}}

\newcommand{\Xc}{\mathcal{X}}

\newcommand{\Kc}{\mathcal{K}}
\newcommand{\Sc}{S^{\sf c}}
\newcommand{\Kcc}{\mathcal{K}^{\sf c}}

\title{Capacity of Coded Index Modulation}
\author{
\IEEEauthorblockN{Lakshmi~Natarajan, Yi~Hong, and Emanuele~Viterbo}%
\IEEEauthorblockA{Department of Electrical \& Computer Systems Engineering\\
Monash University, Clayton, VIC 3800, Australia \\
\{lakshmi.natarajan, yi.hong, emanuele.viterbo\}@monash.edu}%
}

\begin{document}

\maketitle

\begin{abstract}
We consider the special case of index coding over the Gaussian broadcast channel where each receiver has prior knowledge of a subset of messages at the transmitter and demands all the messages from the source. 
We propose a concatenated coding scheme for this problem, using an index code for the Gaussian channel as an inner code/modulation to exploit side information at the receivers, and an outer code to attain coding gain against the channel noise.
We derive the capacity region of this scheme by viewing the resulting channel as a multiple-access channel with many receivers, and relate it to the \emph{side information gain} -- which is a measure of the advantage of a code in utilizing receiver side information -- of the inner index code/modulation.
We demonstrate the utility of the proposed architecture by simulating the performance of an index code/modulation concatenated with an off-the-shelf convolutional code through bit-interleaved coded-modulation. 
\end{abstract}

\section{Introduction} \label{sec:1}

We consider coding schemes for the index coding problem~\cite{YBJK_IEEE_IT_11,ALSWH_FOCS_08,RSG_IEEE_IT_10} over an additive white Gaussian noise~(AWGN) broadcast channel. The sender encodes $K$ independent messages to be broadcast to a set of receivers, each of which demands a subset of the transmit messages while having the prior knowledge of the values of a different subset of messages as side information. The exact capacity region of this Gaussian version of the index coding problem~\cite{Wu_ISIT_07,KrS_ITW_07,SiC_ISIT_14,AOJ_ISIT_14,Tun_IEEE_IT_06}, with general message requests and side informations, is known only for the two receiver case~\cite{Wu_ISIT_07,KrS_ITW_07}.

\let\thefootnote\relax\footnotetext{\copyright~2015 IEEE. Personal use of this material is permitted. Permission from IEEE must
be obtained for all other uses, including reprinting/republishing this material for
advertising or promotional purposes, collecting new collected works for resale or
redistribution to servers or lists, or reuse of any copyrighted component of this work
in other works.} 

In this paper, we consider the special case of Gaussian index coding where every receiver demands all the messages at the source, the capacity of which is given in~\cite{Tun_IEEE_IT_06}. 
Let each receiver be denoted by the tuple $({\sf SNR},S)$, where ${\sf SNR}$ is the receiver signal-to-noise ratio, and $S \subset \{1,\dots,K\}$ is the subset of the indices of the messages known at the receiver. A rate of $R_1,\dots,R_K$ bits per real channel use (i.e., bits per real dimension, denoted~b/dim) is achievable for the $K$ messages if and only if for every receiver $({\sf SNR},S)$ we have~\cite{Tun_IEEE_IT_06}
\begin{align} \label{eq:capacity_unconstrained}
{\textstyle \sfrac{1}{2} \log_2\left( 1 + {\sf SNR} \right) > R - R_S},
\end{align} 
where \mbox{$R=\sum_{k=1}^{K}R_k$} is the sum rate at the source, and \mbox{$R_S=\sum_{k \in S}R_k$} is the \emph{side information rate} at the receiver.
From~\eqref{eq:capacity_unconstrained}, at high message rates, the availability of side information corresponding to $S$ reduces the minimum required ${\sf SNR}$ from approximately $2^{2R}$ to $2^{2(R-R_S)}$, or equivalently, by a factor of \mbox{$10 \log_{10}\left(2^{2R_S}\right) \approx 6R_S$~dB}. Hence, a good index code $\Xc$ is
\begin{inparaenum} [\itshape i\upshape)]
\item a capacity achieving code in the classical point-to-point AWGN channel; and
\item converts each bit per dimension of side information rate into an apparent ${\sf SNR}$ gain of approximately $6$~dB.
\end{inparaenum}

In~\cite{XFKC_CISS_06,BaC_ITW_11,MLV_PIMRC_12}, binary codes were constructed that admit improvement in error performance with the availability of side information at the receivers. 
Index codes for the Gaussian broadcast channel were constructed in~\cite{NHV_arxiv_14} that transform receiver side information into apparent ${\sf SNR}$ gains, and the notion of \emph{side information gain} was introduced as a measure of their efficiency. 
Designing a good index code is equivalent to constructing a code $\Xc$ with a large minimum Euclidean distance, in order to maximize the channel coding gain, and a large side information gain $\Gamma(\Xc)$, to maximize the gain available from utilizing receiver side information.
Although the index codes of~\cite{NHV_arxiv_14} have a large side information gain, they are not efficient against the channel noise. Hence, the known index codes, as such, are not sufficient to achieve near-capacity performance in the Gaussian broadcast channel.

In this paper, we propose a concatenated scheme, called \emph{coded index modulation}, that simultaneously achieves both coding gain and side information gain. In our scheme, the $K$ information streams are encoded independently using strong outer codes, and the resulting codewords are modulated using an inner index code. While the outer codes provide error resilience, the inner index code, henceforth referred to as \emph{index modulation}, allows the receivers to exploit side information.
We derive the capacity region of coded index modulation by viewing the resulting channel as a multiple-access channel with many receivers~\cite{Ulr_InfCont_75}, and relate it to the side information gain of the index modulation (Section~\ref{sec:4}).
We illustrate the simultaneous achievability of both coding gain and side information gain by simulating the performance of a 64-QAM index modulation coded using a 16-state convolutional code through bit-interleaved coded-modulation~\cite{CTB_IT_98,LiR_ElecLet_98} (Section~\ref{sec:5}). The system model is introduced in Section~\ref{sec:2}.

\emph{Notation:} Vectors are denoted by bold lower case letters.
Random variables are denoted by plain upper case letters (eg. $X$), while a particular instance of a random variable is denoted using a lower case font (eg. $x$).
The symbol $S^\mathsf{c}$ denotes the complement of the set $S$, and $\varnothing$ is the empty set. 

\section{Index codes for Gaussian broadcast channel} \label{sec:2}

In this section, we introduce the channel model, and review the notion of \emph{side information gain}~\cite{NHV_arxiv_14} which is a measure of the efficiency of a coding scheme in Gaussian broadcast channels with receiver side information. 

Consider a transmitter with $K$ independent messages $x_1,\dots,x_K$, taking values from finite alphabets $\Xc_1,\dots,\Xc_K$, respectively, to be broadcast to a set of receivers, each of which demands all the messages at the source. 
Let $({\sf SNR},S)$ denote a receiver with signal-to-noise ratio ${\sf SNR}$ and the prior knowledge of the subset of information symbols \mbox{$\pmb{x}_S \triangleq (x_k,k \in S)$}, for some \mbox{$S \subset \{1,\dots,K\}$}, as side information.

An $n$-dimensional \emph{index code} $(\rho,\Xc)$ for $\Xc_1,\dots,\Xc_K$ consists of a codebook \mbox{$\Xc \subset \Rb^n$}, and an encoding function \mbox{$\rho: \Xc_1 \times \cdots \times \Xc_K \to \Xc$}, where \mbox{$\pmb{x}=\rho(x_1,\dots,x_K)$} is the transmit signal. 
The resulting spectral efficiency for the $k^{\text{th}}$ message is $R_k = \sfrac{1}{n} \, \log_2 |\Xc_k|$ bits per dimension (b/dim). 
The encoding operation is independent of the number of receivers and the side information available to each of them. 
Indeed, the capacity-achieving scheme of~\cite{Tun_IEEE_IT_06} does not utilize this information at the encoder.

Given the channel output \mbox{$\pmb{y} = \pmb{x} + \pmb{z}$}, where $\pmb{z}$ is the additive white Gaussian noise, and the side information \mbox{$\pmb{x}_S=\pmb{a}_S$}, i.e., \mbox{$x_k=a_k$} for \mbox{$k \in S$}, the maximum-likelihood decoder at the receiver $({\sf SNR},S)$ restricts its choice of transmit vectors to the subcode \mbox{$\Xc_{\pmb{a}_S} \subset \Xc$} obtained by expurgating all the codewords corresponding to \mbox{$\pmb{x}_S \neq \pmb{a}_S$}.
Decoding $\pmb{y}$ to the subcode $\Xc_{\pmb{a}_S}$, instead of $\Xc$, may improve the minimum distance between the valid codewords, and hence the error performance of the receiver $({\sf SNR},S)$ over a receiver with no side information. 
Let $d_0$ be the minimum Euclidean distance between any two vectors in $\Xc$, $d_{\pmb{a}_S}$ be the minimum distance of the codewords in $\Xc_{\pmb{a}_S}$, and $d_S$ be the minimum value of $d_{\pmb{a}_S}$ over all possible values of $\pmb{a}_S$.
At high ${\sf SNR}$, the side information corresponding to $S$ provides an ${\sf SNR}$ gain of approximately $10\log_{10}\left( \sfrac{d_S^2}{d_0^2} \right)$~dB. Normalizing by the \emph{side information rate}~\cite{NHV_arxiv_14} \mbox{$R_S \triangleq \sum_{k \in S}R_k$}, we see that each bit per dimension of side information provides an apparent gain of $\sfrac{1}{R_S} \, {10 \log_{10}\left(\sfrac{d_S^2}{d_0^2} \right)}$~dB.
We are interested in coding schemes that provide large ${\sf SNR}$ gains for every 
\mbox{$S \subset \{1,\dots,K\}$}.

\begin{definition}[\cite{NHV_arxiv_14}]
The \emph{side information gain} of an index code $(\rho,\Xc)$ is 
\begin{align} \label{eq:Gamma}
\Gamma(\Xc) \triangleq \min_{S \subset \{1,\dots,K\}} \frac{10 \log_{10} \left( \sfrac{d_S^2}{d_0^2} \right)}{R_S} \text{ dB/b/dim}.
\end{align} 
\end{definition}

By using the normalizing factor $R_S$ in~\eqref{eq:Gamma}, the Euclidean distance gain is measured with reference to the amount of side information available at a receiver.
The asymptotic ${\sf SNR}$ gain due to the prior knowledge of $\pmb{x}_S$ is at least \mbox{$\Gamma(\Xc) \times R_S$~dB}, and a large value of $\Gamma(\Xc)$ simultaneously maximizes this gain for every choice of \mbox{$S \subset \{1,\dots,K\}$}.
Note that $\Gamma$ is a relative measure of the performance of the index code in a broadcast channel with receiver side information, computed with respect to the baseline performance of $\Xc$ in a point-to-point AWGN channel with no side information at the receiver ($S = \varnothing$).

\begin{example} \label{ex:16QAM_1}
\begin{figure}[!t]
\centering
\includegraphics[totalheight=2in,width=2in]{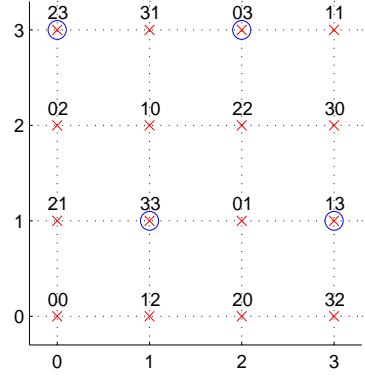}
\caption{The constellation points (crosses) of the $16$-QAM index code labelled using the tuple of input symbols $(x_1,x_2)$. The four points marked with circles constitute the subcode corresponding to the side information \mbox{$x_2=3$}.
}
\label{fig:16QAM_labels}
\vspace{-3mm}
\end{figure} 
Let \mbox{$K=2$}, \mbox{$\Xc_1=\Xc_2=\{0,1,2,3\}$}, and \mbox{$\Xc \subset \Rb^2$} be the $16$-QAM constellation. The dimension of the code $\Xc$ is \mbox{$n=2$}, and the two sources have message rates \mbox{$R_1=R_2=1$~b/dim}. Consider the encoder \mbox{$\rho: \Xc_1 \times \Xc_2 \to \Xc$} given by
\begin{align} \label{eq:ex:16QAM_1:encoding}
\pmb{x} = \rho(x_1,x_2) = \left(x_1 + 2x_2,2x_1 + x_2\right) {\rm~mod}~4, 
\end{align} 
where the ${\rm~mod}~4$ operation is performed component-wise.
Fig.~\ref{fig:16QAM_labels} shows the transmit constellation $\Xc$ where the signal points are labelled with the values of the corresponding tuple of information symbols \mbox{$(x_1,x_2)$}. The minimum Euclidean distance between any two distinct points in $\Xc$ is \mbox{$d_0=1$}. Now suppose the side information available at a receiver with \mbox{$S=\{2\}$} is \mbox{$x_2=3$}. 
From~\eqref{eq:ex:16QAM_1:encoding}, the set of Euclidean coordinates of all the codewords with \mbox{$x_2=3$} is \mbox{$\{(0,3),(1,1),(2,3),(3,1)\} \subset \Xc$}, with the corresponding minimum Euclidean distance $2$. 
Hence, the availability of the side information \mbox{$x_2=3$} increases the minimum distance between the valid codewords to $2$ from \mbox{$d_0=1$}.
Similarly, using direct computation, we obtain \mbox{$d_S=2$} for both \mbox{$S=\{1\}$} and \mbox{$S=\{2\}$}. 
From~\eqref{eq:Gamma}, the side information gain of this code is
 $\Gamma = {10 \log_{10}\left(\sfrac{2^2}{1^2}\right)} \approx 6$~dB/b/dim.
Fig.~\ref{fig:comparison_all} in Section~\ref{sec:5} includes the simulated error performance of $\Xc$ versus ${\sf SNR}$ for three receivers, corresponding to \mbox{$S=\{1\},\{2\}$} and $\varnothing$, respectively.
From Fig.~\ref{fig:comparison_all}, the prior knowledge of either $x_1$ or $x_2$ at a receiver provides an ${\sf SNR}$ gain of approximately $6.5$~dB over \mbox{$S=\varnothing$}, which is consistent with the squared distance gain of $\Gamma(\Xc) \times R_S \approx 6$~dB for $S=\{1\},\{2\}$.\hfill\IEEEQED
\end{example}

\section{Capacity of coded index modulation} \label{sec:4}

The index codes constructed in~\cite{NHV_arxiv_14} have large values of $\Gamma$, and hence, are efficient in exploiting the side information available at the receivers.
However, as in Example~\ref{ex:16QAM_1}, the transmit codebook $\Xc$ may not be an efficient channel code in the traditional single-user AWGN channel with no receiver side information (see~Fig.~\ref{fig:comparison_all}). 
Hence, the index codes of~\cite{NHV_arxiv_14}, as such, may be inadequate to achieve near-capacity performance in the broadcast channel with receiver side information.

We improve the coding gain against channel noise by coding the $K$ message streams independently using strong outer codes over the alphabets $\Xc_1,\dots,\Xc_K$, respectively, and concatenating the encoded streams in the signal space using the  encoding map $\rho$. 
This is illustrated in Fig.~\ref{fig:outer_coded_channel}, where $W_1,\dots,W_K$ are the information symbols, $\mathcal{E}_1,\dots,\mathcal{E}_K$ are the channel encoders, and $X_1,\dots,X_K$ are their coded outputs at a given time instance, which are jointly modulated by $\rho(\cdot)$ into a transmit symbol $X$. The symbol $Y_S$ denotes the channel output at the receiver `${\sf Rx}_S$' which has side information \mbox{$W_S=(W_k,k \in S)$}.
If the $K$ outer codes have a minimum Hamming distance of $d_{H}$ over the alphabets $\Xc_1,\dots,\Xc_K$, respectively, then the minimum squared Euclidean distance between the valid codewords is at least \mbox{$d_H \times d_S^2$} at the receiver ${\sf Rx}_S$.
While the outer codes provide coding gain against channel noise, the \emph{index modulation} $(\rho,\Xc)$ ensures that the receiver performance improves with the availability of side information.
In order to measure the efficiency of this coded index modulation, 
we derive its capacity region, and investigate its dependence on the side information gain $\Gamma(\Xc)$.

\begin{figure}[!t]
\centering
\includegraphics[totalheight=1.125in,width=3.7in]{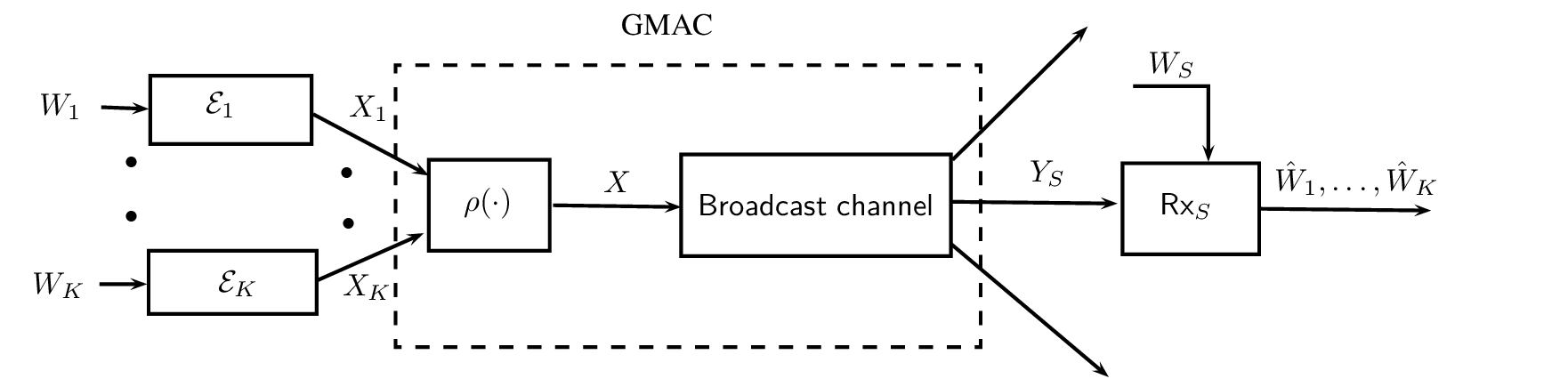}
\caption{Broadcast channel with coded index modulation.}
\label{fig:outer_coded_channel}
\vspace{-5mm}
\end{figure}

\subsection{Capacity region}


If the broadcast channel contains more than one receiver ${\sf Rx}_S$ with the same side information index set $S$, we only consider the node with the least ${\sf SNR}$ among all such receivers. 
We assume the most restrictive scenario, with \mbox{$2^K-1$} receivers, one for each possible side information configuration \mbox{$S \subset \{1,\dots,K\}$}, including \mbox{$S=\varnothing$} and excepting \mbox{$S=\{1,\dots,K\}$}. 
The ${\sf SNR}$ at ${\sf Rx}_{S}$ is denoted by ${\sf SNR}_S$.
The following working assumption, used in the proof of Theorem~\ref{thm:cap}, will simplify the analysis by reducing the number of hyperplanes that define the capacity region:
\begin{align} \label{eq:working_assumption}
{\sf SNR}_{S} \geq {\sf SNR}_{S'} \text{ for every } S \subset S'.
\end{align} 
Since ${\sf Rx}_{S}$ has less side information than ${\sf Rx}_{S'}$, it has a higher minimum ${\sf SNR}$ requirement for any given tuple of achievable message rates $(R_1,\dots,R_K)$. 
Hence, it is reasonable to assume that ${\sf Rx}_{S}$ is operating at a higher ${\sf SNR}$ than ${\sf Rx}_{S'}$.

A \emph{general multiple access channel~(GMAC)}~\cite{Ulr_InfCont_75} with $K$ transmitters and $L$ receivers consists of a memoryless channel $p(\tilde{y}_1,\dots,\tilde{y}_L|x_1,\dots,x_K)$ with inputs $X_1,\dots,X_K$ and channel outputs $\widetilde{Y}_1,\dots,\widetilde{Y}_L$. 
The broadcast channel with coded index modulation, as shown in Fig.~\ref{fig:outer_coded_channel}, can be viewed as a GMAC by considering the $K$ information sources as $K$ independent transmitters, and the function $\rho$ as a part of the physical channel.
The resulting GMAC has \mbox{$L=2^K-1$} receivers that are indexed by \mbox{$S$}.
Receiver side information can be absorbed in the GMAC channel model by setting
\begin{align} \label{eq:Y_tilde}
\widetilde{Y}_S \triangleq (Y_S,X_S),
\end{align} 
where \mbox{$X_S \triangleq (X_k,k \in S)$}, i.e., by considering the side information as part of the channel output.
We now recall the capacity region~\cite{Ulr_InfCont_75} of a GMAC using the current notation.
Let the random variables \mbox{$X_1,\dots,X_K$} be independently distributed on the $n$-dimensional alphabets $\Xc_1,\dots,\Xc_K$ with probability distributions $p(x_1),\dots,p(x_K)$, respectively.

\begin{theorem}[\cite{Ulr_InfCont_75}] \label{thm:GMAC}
For a given set of input distributions $p(x_1),\dots,p(x_K)$, 
the capacity region of the GMAC 
is the set of all rate tuples $(R_1,\dots,R_K)$ satisfying
\begin{align*}
\sum_{k \in \Kc} R_k \leq \frac{1}{n} I(X_{\Kc};\widetilde{Y}_{S}|X_{\Kcc}) \text{ for all } \Kc,S \subset \{1,\dots,K\}.
\end{align*}
\end{theorem}

Using~\eqref{eq:working_assumption} and~\eqref{eq:Y_tilde} we apply Theorem~\ref{thm:GMAC} to coded index modulation, and simplify the capacity region by reducing the number of constraints from $(2^K-1)^2$ to $2^K-1$.

\begin{theorem} \label{thm:cap}
Let the input distributions $p(x_1),\dots,p(x_K)$ be fixed.
A rate tuple $(R_1,\dots,R_K)$ is achievable using coded index modulation over a broadcast channel (with receiver side information) satisfying~\eqref{eq:working_assumption} if and only if 
\begin{align} \label{eq:capacity_under_assumption}
\sum_{k \in \Sc} R_k \leq \frac{1}{n} I(X_{\Sc};Y_S|X_S) \text{ for all } S \subset \{1,\dots,K\}.
\end{align} 
\end{theorem}
\begin{proof}
Using~\eqref{eq:Y_tilde} in Theorem~\ref{thm:GMAC}, the capacity region is the set of all rate tuples satisfying
\begin{align}
n \sum_{k \in \Kc} R_k &\leq I(X_{\Kc};Y_S,X_S|X_{\Kcc}) \label{eq:thm:0} \\
   &= I(X_{\Kc};X_S|X_{\Kcc}) + I(X_{\Kc};Y_S|X_S,X_{\Kcc}), \label{eq:thm:1}            
\end{align} 
for every \mbox{$S,\Kc \subset \{1,\dots,K\}$}. 
The inequality~\eqref{eq:capacity_under_assumption} is the same as~\eqref{eq:thm:0} when \mbox{$\Kc=\Sc$}.
In order to prove the theorem it is enough to show that~\eqref{eq:capacity_under_assumption} implies~\eqref{eq:thm:0} for any other choice of $(\Kc,S)$.

We obtain \mbox{$nR_k \leq I(X_k;Y_{\{k\}^{\sf c}}|X_{\{k\}^{\sf c}}) \leq H(X_k)$} by using \mbox{$S=\{k\}^{\sf c}$} in~\eqref{eq:capacity_under_assumption}. 
Hence,~\eqref{eq:capacity_under_assumption} implies
\begin{align}
\sum_{k \in \Kc} nR_k \leq \sum_{k \in \Kc} H(X_k) \text{ for all } \Kc \subset \{1,\dots,K\}. \label{eq:thm:1.5}
\end{align} 

We now simplify each of the two terms in~\eqref{eq:thm:1} for arbitrary $\Kc$ and $S$. 
Utilizing the fact that $X_1,\dots,X_K$ are independent random variables, 
we have
\begin{align}
I(X_{\Kc};X_S|X_{\Kcc}) &= H(X_{\Kc}|X_{\Kcc}) - H(X_{\Kc}|X_S,X_{\Kcc}) \nonumber \\ 
&= H(X_{\Kc}) - H(X_{\Kc}|X_S) \nonumber \\
&= H(X_{\Kc}) - \left( H(X_{\Kc}) - H(X_{S \cap \Kc}) \right) \nonumber \\
&= H(X_{S \cap \Kc}) = \sum_{k \in S \cap \Kc} H(X_k). \label{eq:thm:2}
\end{align} 
The second term of~\eqref{eq:thm:1} can be rewritten as
\begin{align}
I(&X_{\Kc};Y_S|X_{S \cup \Kcc}) \nonumber \\
&= I(X_{\Kc \cap \Sc},X_{\Kc \cap S};Y_S|X_{S \cup \Kcc}) \nonumber \\
&= I(X_{\Kc \cap \Sc};Y_S|X_{S \cup \Kcc}) + I(X_{\Kc \cap S};Y_S|X_{S \cup \Kcc},X_{\Kc \cap \Sc}) \nonumber \\
&= I(X_{\Kc \cap \Sc};Y_S|X_{S \cup \Kcc}), \label{eq:thm:3}
\end{align} 
where the last equality follows from the fact that \mbox{$I(X_{\Kc \cap S};Y_S|X_{S \cup \Kcc},X_{\Kc \cap \Sc})=0$}, since \mbox{$(\Kc \cap S) \subset (S \cup \Kcc)$}. 
Substituting~\eqref{eq:thm:2} and~\eqref{eq:thm:3} in~\eqref{eq:thm:1},
\begin{align} \label{eq:thm:4}
n \sum_{k \in \Kc} R_k &\leq \sum_{k \in S \cap \Kc} H(X_k) + I(X_{\Kc \cap \Sc};Y_S|X_{S \cup \Kcc}).
\end{align}
Rewriting $\sum_{k \in \Kc}R_k$ as the sum $\sum_{k \in \Kc \cap S}R_k + \sum_{k \in \Kc \cap \Sc}R_k$, we observe that~\eqref{eq:thm:4} is the sum of 
\begin{align} 
n \sum_{k \in \Kc \cap S} R_k &\leq \sum_{k \in \Kc \cap S} H(X_k), \label{eq:thm:5} \text{ and} \\ 
n \sum_{k \in \Kc \cap \Sc} R_k &\leq I(X_{\Kc \cap \Sc};Y_S|X_{\Kcc \cup S}). \label{eq:thm:6}
\end{align} 
Since~\eqref{eq:thm:5} follows from~\eqref{eq:capacity_under_assumption} by means of~\eqref{eq:thm:1.5}, we are left with showing that~\eqref{eq:thm:6} follows from~\eqref{eq:capacity_under_assumption}.

Using~\eqref{eq:capacity_under_assumption} with the index set \mbox{$\Kcc \cup S$} instead of $S$, and the working assumption~\eqref{eq:working_assumption} that $Y_{\Kcc \cup S}$ is a stochastically degraded~\cite{CoT_JohnWiley_12} version of $Y_{S}$ (because of the lower ${\sf SNR}$),
\begin{align*}
n\sum_{k \in \Kc \cap \Sc}R_k &\leq I(X_{\Kc \cap \Sc};Y_{\Kcc \cup S}|X_{\Kcc \cup S}) \\
                              &\leq I(X_{\Kc \cap \Sc};Y_S|X_{\Kcc \cup S}).
\end{align*} 
This completes the proof.
\end{proof}

\begin{example} \label{ex:64QAM}

\begin{figure}[!t]
\centering
\psfrag{xaxis}[][][0.75]{${\sf SNR}$ (in dB)}
\psfrag{yaxis}[][][0.75]{Mutual information (in b/dim)}
\psfrag{label1111111}[][][0.6]{$~~~~~\sfrac{1}{2}\log_2(1+{\sf SNR})$}
\psfrag{label2222222}[][][0.6]{$~~~\sfrac{1}{2} \, I(X_1,X_2;Y)$}
\psfrag{label3333333}[][][0.6]{$~~~\sfrac{1}{2} \, I(X_1;Y|X_2)$}
\psfrag{label4444444}[][][0.6]{$~~~\sfrac{1}{2} \, I(X_2;Y|X_1)$}
\includegraphics[totalheight=2.4in,width=3.0in]{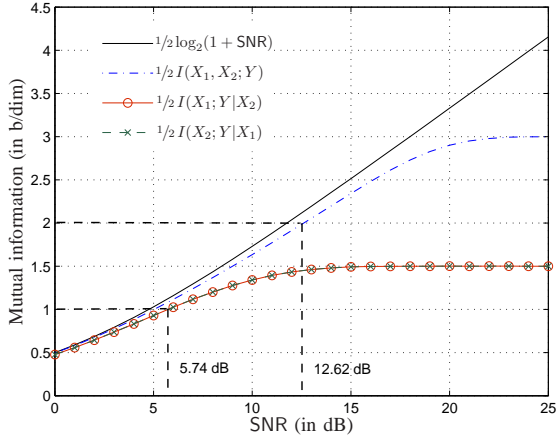}
\caption{Mutual information rates available from using the $64$-QAM index modulation of Example~\ref{ex:64QAM}.}
\label{fig:64QAM_achievable_rate}
\vspace{-3mm}
\end{figure}

Consider \mbox{$K=2$} messages, each of rate $1$~b/dim, to be transmitted across a broadcast channel consisting of \mbox{$2^K-1=3$} receivers with channel outputs $Y_{\varnothing},Y_1,Y_2$ and side information index sets \mbox{$S=\varnothing$}, \mbox{$S=\{1\}$}, \mbox{$S=\{2\}$}, respectively. 
Let \mbox{$\Xc_1=\Xc_2=\{0,1,\dots,7\}$}, and $(\rho,\Xc)$ be the $64$-QAM index modulation defined by the map
$\pmb{x} = \rho(x_1,x_2) = (x_1+2x_2,2x_1+x_2) {\rm~mod}~8$.
By direct computation, we obtain \mbox{$d_0=1$}, and \mbox{$d_S=\sqrt{5}$} for both $S=\{1\},\{2\}$, resulting in \mbox{$\Gamma=4.66$~dB/b/dim}.
Let $X_1$ and $X_2$ be independently and uniformly distributed over $\Xc_1$ and $\Xc_2$, respectively, and \mbox{$X=\rho(X_1,X_2)$}.
From Theorem~\ref{thm:cap}, the rate tuple $(R_1,R_2)$ is achievable with the modulation $(\rho,\Xc)$ if and only if
\begin{align}
&R_1 \leq \sfrac{1}{2} \, I(X_1;Y_2|X_2), 
R_2 \leq \sfrac{1}{2} \, I(X_2;Y_1|X_1), \nonumber \\
&~~~~~~~~~~~~~R_1 + R_2 \leq \sfrac{1}{2} \, I(X_1,X_2;Y_{\varnothing}). \label{eq:ex:inner_code_limit}
\end{align} 
On the other hand, the capacity region with Gaussian input distribution and joint encoding of the sources is~\cite{Tun_IEEE_IT_06}
\begin{align}
&R_1 \leq \sfrac{1}{2} \, \log_2\left( 1 + {\sf SNR}_2\right),
R_2 \leq \sfrac{1}{2} \, \log_2\left( 1 + {\sf SNR}_1\right), \nonumber \\
&~~~~~~~~~~~~R_1 + R_2 \leq \sfrac{1}{2} \, \log_2\left( 1 + {\sf SNR}_{\varnothing}\right). \label{eq:ex:abs_capacity_limit}
\end{align}
The mutual information rates~\eqref{eq:ex:inner_code_limit}, obtained from Monte-Carlo methods, are shown in Fig.~\ref{fig:64QAM_achievable_rate} for a generic channel output $Y$ as a function of ${\sf SNR}$. The figure also shows the curve $\sfrac{1}{2} \, \log_2\left( 1 + {\sf SNR}\right)$ that dictates the absolute capacity~\eqref{eq:ex:abs_capacity_limit} achievable with Gaussian input distribution.
From Fig.~\ref{fig:64QAM_achievable_rate}, we observe that, to achieve \mbox{$(R_1,R_2)=(1,1)$} using the index modulation $(\rho,\Xc)$, the ${\sf SNR}$ requirement at ${\sf Rx}_1$, ${\sf Rx}_2$ and ${\sf Rx}_{\varnothing}$ are $5.74$~dB, $5.74$~dB and $12.62$~dB, respectively. These are within $1$~dB of the absolute ${\sf SNR}$ limits, viz. $4.77$~dB, $4.77$~dB and $11.77$~dB, obtained from~\eqref{eq:ex:abs_capacity_limit}. \hfill\QED
\end{example}

\subsection{Dependence on the side information gain $\Gamma(\Xc)$}

We use Fano's inequality to relate the information rates \mbox{$\sfrac{1}{n} \, I(X_{\Sc};Y_S|X_S)$}, that define the capacity region (cf.~Theorem~\ref{thm:cap}), to the side information gain $\Gamma$ of the modulation scheme used.
Larger values of $I(X_{\Sc};Y_S|X_S)$ imply a larger capacity region.
Let $X_1,\dots,X_K$ be uniformly distributed on $\Xc_1,\dots,\Xc_K$, respectively, and ${\sf P_e}\left(X_{\Sc}|Y_S,X_S\right)$ be the probability of error of the optimal decoder that decodes $X_{\Sc}$ given the values of $X_S$ and  $Y_S$. 
Using Fano's inequality~\cite{CoT_JohnWiley_12}
\begin{equation} \label{eq:Fanos}
H(X_{\Sc}|Y_S,X_S) \leq 1 + {\sf P_e}(X_{\Sc}|Y_S,X_S) \, \log_2 \prod_{k \in \Sc}|\Xc_k|,
\end{equation} 
we obtain the following lower bound,
\begin{align}
I(X_{\Sc};Y_S|X_S) &= H(X_{\Sc}|X_S) - H(X_{\Sc}|Y_S,X_S) \nonumber \\
                   &= \log_2 \prod_{k \in \Sc}|\Xc_k| - H(X_{\Sc}|Y_S,X_S) \nonumber \\
                   &\geq  \left( 1 - {\sf P_e}(X_{\Sc}|Y_S,X_S)\right) \log_2 \prod_{k \in \Sc}|\Xc_k|- 1. \nonumber 
\end{align} 
Thus, from Theorem~\ref{thm:cap} and the above inequality, smaller values of ${\sf P_e}(X_{\Sc}|Y_S,X_S)$, imply larger lower bounds on the achievable rates.
%
For fixed values of ${\sf SNR}_S$, \mbox{$S \subset \{1,\dots,K\}$}, maximizing $\Gamma$ maximizes $d_S$ simultaneously for all \mbox{$S \neq \varnothing$}, and hence minimizes ${\sf P_e}(X_{\Sc}|Y_S,X_S)$ for all \mbox{$S \neq \varnothing$}. 
We thus expect that larger values of $\Gamma$ will lead to higher achievable rates for given $\left({\sf SNR}_S,S \subset \{1,\dots,K\}\right)$, or equivalently, lower minimum ${\sf SNR}$ requirements for a given tuple of message rates $(R_1,\dots,R_K)$.

\begin{example}
We consider two different index modulations over $256$-QAM with \mbox{$K=n=2$} and \mbox{$\Xc_1=\Xc_2=\{0,1,\dots,15\}$}. 
The two schemes, corresponding to the encoding functions \mbox{$\rho(x_1,x_2)=(x_1+12x_2,12x_1+x_2) {\rm~mod}~16$}
and \mbox{$\rho(x_1,x_2)=(x_1+2x_2,2x_1+x_2) {\rm~mod}~16$}, respectively, have 
\mbox{$\Gamma=6.02$~dB/b/dim} and \mbox{$\Gamma=3.49$~dB/b/dim}.
%
%
Assume uniform distribution of $X_1,X_2$ on $\Xc_1,\Xc_2$, respectively.
Using Monte-Carlo methods, we obtain the minimum required ${\sf SNR}$ at ${\sf Rx}_1,{\sf Rx}_2$ and ${\sf Rx}_{\varnothing}$, for each of the two modulation schemes to support \mbox{$R_1=R_2=1.5$}~b/dim. 
The minimum requirements for the two schemes on $({\sf SNR}_1,{\sf SNR}_2,{\sf SNR}_{\varnothing})$ are $(9.5,9.5,19.2)$ and $(11.3,11.3,19.2)$, in~dB, respectively. The first scheme, whose side information gain is larger by a factor of $2.53$~dB/b/dim, can tolerate $1.8$~dB of additional noise at ${\sf Rx}_1$ and ${\sf Rx}_2$. \hfill\QED
\end{example}

\section{Simulation Results and Conclusion} \label{sec:5}

\begin{figure}[!t]
\centering
\psfrag{xaxis}[ct][][0.75]{${\sf SNR}$~(in dB)}
\psfrag{yaxis}[cb][][0.75]{{\rm Bit~error~rate}}
\psfrag{label11111}[][][0.6]{$\, S = \{1\}$}
\psfrag{label22222}[][][0.6]{$\, S = \{2\}$}
\psfrag{label33333}[][][0.6]{$\, S = \varnothing$}
\psfrag{C1}[][][0.75][90]{{\rm Capacity limit} $S=\{1\},\{2\}$}
\psfrag{C2}[][][0.75][90]{{\rm Capacity limit} $S=\varnothing$}
\includegraphics[totalheight=2.5in,width=3.4in]{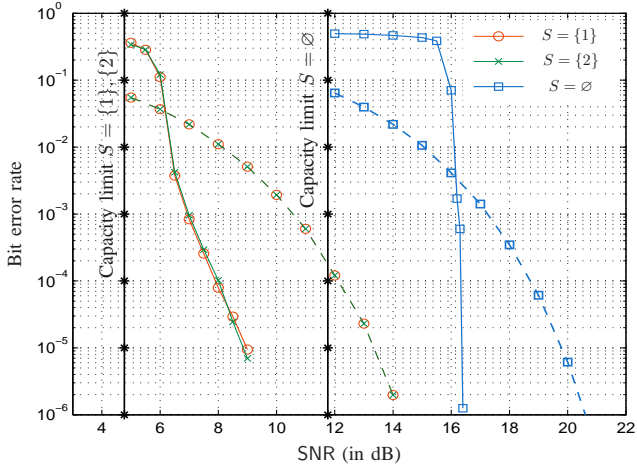}
\caption{
Dashed lines represent $16$-QAM index code of Example~\ref{ex:16QAM_1} without any outer code, solid lines represent $64$-QAM index code of Example~\ref{ex:64QAM} with a rate-\sfrac{2}{3} convolutional code as outer code.}
\label{fig:comparison_all}
\vspace{-3mm}
\end{figure} 

In this section, we demonstrate the utility of coded index modulation in a Gaussian broadcast channel using simulations.
We consider \mbox{$K=2$} messages with \mbox{$R_1=R_2=1$}~b/dim, and $3$ receivers corresponding to $S=\{1\},\{2\}$ and $\varnothing$, respectively.
We present the bit error rate performance for two encoding schemes:
\begin{inparaenum}[\itshape i\upshape)]
\item the $16$-QAM index modulation of Example~\ref{ex:16QAM_1} without any outer code; and
\item the $64$-QAM index modulation of Example~\ref{ex:64QAM} with bit-interleaved coded modulation (BICM)~\cite{CTB_IT_98} using a rate-$\sfrac{2}{3}$, $16$-state convolutional code~\cite{DMW_IT_82}.
\end{inparaenum}
In the first scheme, for \mbox{$k=1,2$}, two information bits from the $k^{\text{th}}$ source are mapped to $\Xc_k$, 
and maximum-likelihood decoding is performed at the receivers.
In the second scheme, a block of $3996$ information bits from the $k^{\text{th}}$ source is encoded using the terminated convolutional code, the resulting $6000$ coded bits are interleaved, and then mapped to the constellation $\Xc_k$, three bits at a time. 
The receivers perform iterative decoding of BICM~\cite{LiR_ElecLet_98} using three soft-input soft-output decoding modules~\cite{BDMP_CommLet_97}: a soft demapper for the $64$-QAM constellation, and two BCJR decoders~\cite{BCJR_IT_74}, one for each convolution coded block. In each iteration, the two BCJR modules exchange extrinsic probabilities through the demapper. The decision on information bits is made after $8$ iterations.

The bit error rate performance of both the schemes, for all three receivers, is shown in Fig.~\ref{fig:comparison_all}. Also shown are the capacity limits~\eqref{eq:ex:abs_capacity_limit} on the ${\sf SNR}$ for \mbox{$R_1=R_2=1$}~b/dim.
At bit error rate $10^{-5}$, the availability of side information provides an apparent ${\sf SNR}$ gain of $6.5$~dB and $7.4$~dB in the two coding schemes, respectively. Further, the BICM-coded system is $4.6$~dB and $4.2$~dB away from capacity for \mbox{$S=\varnothing$} and \mbox{$S=\{1\},\{2\}$}, respectively, and has gained by $3.4$~dB and $4.3$~dB over the uncoded $16$-QAM scheme.

In this paper, we have proposed coded index modulation that separates the problem of coding for utilizing receiver side information from that of coding against channel noise.
This transforms the problem of designing good index codes into two separate problems, viz. constructing index modulations with large side information gain, and designing good channel codes for a noisy multiple access channel. We derived the capacity region of coded index modulation, and demonstrated the potential of this scheme through simulations.

We have shown that index modulations with larger side information gains lead to larger lower bounds on achievable rates. It will be interesting to derive explicit bounds that show that coded index modulation can approach the index coding capacity. 
While we relied on BICM for our simulation results, designing good outer codes that are matched to the index modulations is yet to be addressed. 


\end{document}